\documentclass{amsart}
\usepackage[margin=1in]{geometry}
\usepackage{amsfonts}
\usepackage{amssymb}
\usepackage[dvips]{graphics}
\usepackage{epsfig}
\usepackage{euscript}
\usepackage{color}

\newcommand{\N}{\mathbb{N}}
\newtheorem{theorem}{Theorem}[section]
          
\newtheorem{lemma}[theorem]{Lemma}              
              
\newtheorem{proposition}[theorem]{Proposition}

\newtheorem{remark}[theorem]{Remark}
\newcommand{\rem}[1]{{\bf Remark:}}

\newcommand{\Section}[1]{\setcounter{equation}{0}\section{#1}}

\renewenvironment{proof}{\noindent {\bf Proof: }}{\QED\medskip}
\def\QED{{\hspace*{\fill}{\vrule height 1ex width 1ex }\quad} 
 \vskip 0pt plus20pt}
\newenvironment{proofof}[1]{\noindent {\bf Proof of #1:}}{\QED\medskip}
\newcommand{\be}{\begin{equation}}
\newcommand{\ee}{\end{equation}}
\newcommand{\bea}{\begin{eqnarray}}
\newcommand{\eea}{\end{eqnarray}}
\newcommand{\beann}{\begin{eqnarray*}}
\newcommand{\eeann}{\end{eqnarray*}}

\newcommand{\ket}[1]{\vert{#1}\rangle}

\newcommand{\Hil}{\mathcal{H}}

\newcommand{\Spin}{J}

\newcommand{\J}{J}
\renewcommand{\a}{\alpha}
\newcommand{\D}{\Delta}
\newcommand{\bm}{\vec{m}}
\begin{document}

\begin{center}
{\LARGE \bf Isolated Eigenvalues of the Ferromagnetic Spin-$\Spin$ XXZ Chain with Kink Boundary Conditions\\[27pt]}
{\large \bf Jaideep  Mulherkar$^a$, Bruno Nachtergaele$^a$, Robert Sims$^a$, Shannon Starr$^b$\\[10pt]}
\end{center}
$^a$Department of Mathematics,
University of California, Davis,
Davis, CA 95616-8633, USA\\[10pt]
$^b$Department of Mathematics,
University of Rochester,
Rochester, NY 14627, USA\\[10pt]
\email{jmulherkar@math.ucdavis.edu, bxn@math.ucdavis.edu,
rjsims@math.ucdavis.edu,\newline
sstarr@math.rochester.edu}\\[10pt]

\noindent
{\bf Abstract:}

\noindent
We investigate the low-lying excited states of the spin $\Spin$ ferromagnetic XXZ 
chain with Ising anisotropy $\Delta$ and kink boundary conditions. 
Since the third component of the total magnetization, $M$, is conserved, it is
meaningful to study the spectrum for each fixed value of $M$.
We prove that for $J\geq 3/2$ the lowest excited eigenvalues are
separated by a gap from the rest of the spectrum, uniformly in the
length of the chain. 
In the thermodynamic limit, this means that there are a positive
number of 
excitations above the ground state and below the essential spectrum.

\vspace{8pt}
\newline
{\small \bf MCS numbers:} 81Q15, 82B10, 82B24, 82D40
\vfill
\hrule width2truein \smallskip {\baselineskip=10pt \noindent Copyright
\copyright\ 2007 by the authors. Reproduction of this article in its entirety, 
by any means, is permitted for non-commercial purposes.\par }

\newpage

\Section{Introduction}

In this paper, we are investigating the existence of isolated
excited states in certain one-dimensional, quantum spin
models of magnetic systems. It turns out that if the spins
are of magnitude $1$ or more and their interactions have a suitable
anisotropy, such as in the ferromagnetic XXZ Heisenberg model,
isolated excited states are possible. For the spin $1/2$ 
chain the ground states are separated by a gap to the rest
of the spectrum, and there are no isolated eigenvalues below the
continuum. 

Our main result is a mathematical demonstration that such states 
indeed exist for sufficiently large anisotropy. 
Concretely, we study the one-dimensional spin $J$ ferromagnetic XXZ model with the following
boundary terms. The Hamiltonian is
$$
H_L^{\rm k}(\Delta^{-1})= \sum_{\alpha=-L}^{L-1} \Big[(\Spin^2-S_\alpha^3 S_{\alpha+1}^3) - \Delta^{-1}(S_\alpha^1 S_{\alpha+1}^1 + S_\alpha^2 S_{\alpha+1}^2)\Big]+
    \Spin \sqrt{1-\Delta^{-2}}(S_{-L}^3 - S_L^3)
$$
where $S_\alpha^1$,$S_\alpha^2$ and $S_\alpha^3$ are the spin $\Spin$ matrices acting on the site $\alpha$. Apart from the magnitude of 
the spins, $J$, the main parameter of the model is the anisotropy $\Delta>1$, 
and we will refer to the limit $\Delta\rightarrow\infty$
as the \emph{Ising limit}. In the case of $J=1/2$ these boundary
conditions were first introduced in \cite{PS1990}. They lead to ground states with a domain wall
between down spins on the left portion of the chain and up spins on the right. The domain wall is
exponentially localized \cite{bolina2000}. 
The third component of the magnetization, $M$, is conserved, and there is exactly
one ground state for each value of $M$. Different values of $M$ correspond to different positions of the
domain walls, which in one dimension are sometimes referred to as kinks. 
In \cite{ASW1995} and \cite{GW1995} the ground states for this type boundary conditions
were further analyzed and generalized to higher spin.  
A careful analysis of the Ising limit, see Section \ref{sec:Ising},
reveals that for $J \geq 1$ one or more low-lying excitations,
each with a finite degeneracy, closely resemble the domain wall, i.e.
kinked, ground states, and therefore one should expect them to be
resolute under perturbations. In this paper, we first show that these
states exists and correspond to isolated eigenvalues of the finite
volume XXZ chain with sufficiently strong anisotropy. We illustrate
this feature in Figure \ref{fig:excited_kink}. Moreover, as
consequence of the strong localization near the position of the ground
state kink, these eigenvalues only weakly depend on the distance of
the domain wall to the edges of the chain, and for this reason, we are
next able to demonstrate that they persits even after the
thermodynamic limit. The main difficulty we must overcome corresponds
to the fact that, in the thermodynamic limit, the perturbation of the
entire chain is an unbounded operator, and therefore, the standard,
finite-order perturbation theory is inadequate for a rigorous
argument.

The XXZ kink Hamiltonian commutes with the operator $S^3_{tot} =
\sum_{\alpha = -L}^L S^3_\alpha$. We define $\Hil_M$ to be the
eigenspace of $S^3_{tot}$ with eigenvalue $M \in \{-J(2L +1),\dots, J(2 L+1)\}$.
These subspaces are called ``sectors'', and they are invariant subspaces for $H_L^{\rm k}(\Delta^{-1})$. 

It was shown in \cite{ASW1995,GW1995,Mat1996,KN1998} that for each sector 
there is a unique ground state of $H_L^{\rm k}(\Delta^{-1})$ with eigenvalue 0. Moreover, this ground state, $\psi_M$, is given
by the following expression:
\begin{equation*}
\psi_M = \sum
\bigotimes_{\alpha \in [-L,L]}\binom{2\Spin}{\Spin -m_\alpha}^{1/2} q^{\alpha(\Spin-m_\alpha)}\ket{m_\alpha}_\alpha\, ,
\end{equation*}
where the sum is over all configurations for which $\sum_\alpha m_\alpha = M$
and the relationship between $\Delta>1$ and $q\in(0,1)$ is given by $
\Delta=(q+q^{-1})/2$.
A straightforward calculation shows a sharp transition in the magnetization from fully polarized down at the left to fully polarized up at the right.
 For this reason they are called kink ground states. In \cite{KN1998} Koma and Nachtergaele proved that the kink ground states (as well as their spin-flipped or reflected versions the antikinks) comprise the entire set of ground states for the infinite-volume model, aside from the 2 other ground states: the translation invariant maximally magnetized and minimally magnetized all +J and all -J groundstates. Since the infinite-volume Hamiltonian incorporates all possible limits for all possible boundary conditions, this is a strong {\it a posteriori} justification for choosing the kink boundary conditions. It is worth noting that it has been 
proved for the antiferromagnetic model that no such ground states exist 
\cite{DK2003,Mat2005}.

In \cite{KNS2001}, Koma, Nachtergaele, and Starr showed that there is
a spectral gap above each of the ground states in this model 
for all values of $\Spin$. Based on numerical evidence, they also 
made a conjecture that for $\Spin \ge \frac{3}{2}$ the first excited state of the XXZ model is an isolated eigenvalue, and that the magnitude of
the spectral gap is asymptotically given by $J\gamma(\Delta)$,
where $\gamma(\Delta)$ is an eigenvalue of a particular one-particle
problem. Caputo and Martinelli \cite{CM2003} showed that the gap
is indeed of order $J$.

\begin{figure}
\label{fig:excited_kink}
\begin{center}
\epsfig{file=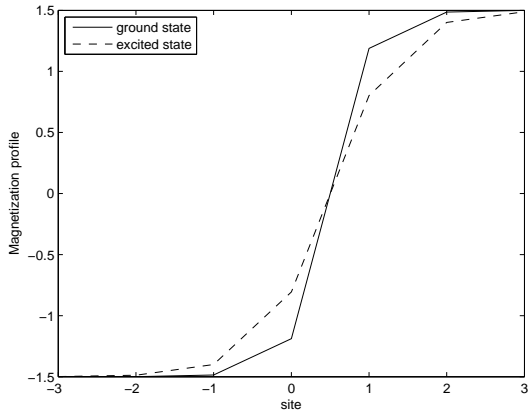,height=6cm,width=8cm}
\caption{The ground state and first excited state of the XXZ chain of length $7$ with $J=\frac{3}{2}$, 
$\Delta=2.5$ in the sector $M=-3/2$.}
\end{center}
\end{figure}

Our main result is a proof that for all sufficiently
large $\Delta$ the first few excitations of the XXZ model are isolated eigenvalues. 
This is true for all $\Spin \geq \frac{3}{2}$ and for spin 1 with $M$ even, which is illustrated in 
Figure ~\ref{fig:xxz_spec} for the spin 3/2 and spin 2 chains.
See Section \ref{sec:Main Results} for the precise
statements. It turns out that in the Ising limit the eigenvalues less than 
$2\Spin$ are all of multiplicity at most $2$ in each sector. Moreover, the first excited states are simple except in the case when 
$J>1$ is an integer and $M=0$ mod $2J$. In this case, they are doubly degenerate. This is discussed in Section \ref{sec:Ising}.
In Section~\ref{sec:mainthm}, we write the XXZ Hamiltonian as an
explicit perturbation of the Ising limit. Theorem~\ref{thm:relbd}
verifies that the perturbation is relatively bounded with respect to
the Ising limit, and we finish this section by demonstrating that our
estimates suffice to guarantee analytic continuation of the limiting 
eigenvalues. It is clear that the same method of proof can 
be applied to other Hamiltonians.

While the question of low-lying excitations is generally interesting,
it is of particular importance in the context of quantum computation.
For quantum computers to become 
a reality we need 
to find or build physical systems that faithfully implement the
quantum gates used in the algorithms of quantum 
computation \cite{Shor1997}. The basic requirement is that the experimenter 
has access to two states of a quantum system that
can be effectively decoupled from environmental noise for a
sufficiently long time, and that transitions between these two
states can be controlled to simulate a number of elementary quantum 
gates (unitary transformations). Systems that have been investigated
intensively are single photon systems, cavity QED, nuclear spins
(using NMR in suitable molecules), atomic levels in ion traps, 
and Josephson rings \cite{NC2000,MOL1999,CVZ1998}.
We believe that if one
could build one-dimensional spin $J$ systems with $J\geq 3/2$, which
interact through an anisotropic interaction such as in the XXZ model,
this would be a good starting point to encode qubits and unitary 
gates. The natural candidates
for control parameters in such systems would be the components of a localized 
magnetic field. From the experimental point of view this is certainly 
a challenging problem. 
This work is a first step toward
developing a mathematical model useful in the
study of optimal control for these systems
such as has already been carried out for nuclear magnetic resonance (NMR)
\cite{KBG2001,MK2005} and superconducting Joshepson qubits \cite{Gal2007}.

\begin{figure}
\label{fig:xxz_spec}
\begin{center}
\epsfig{file=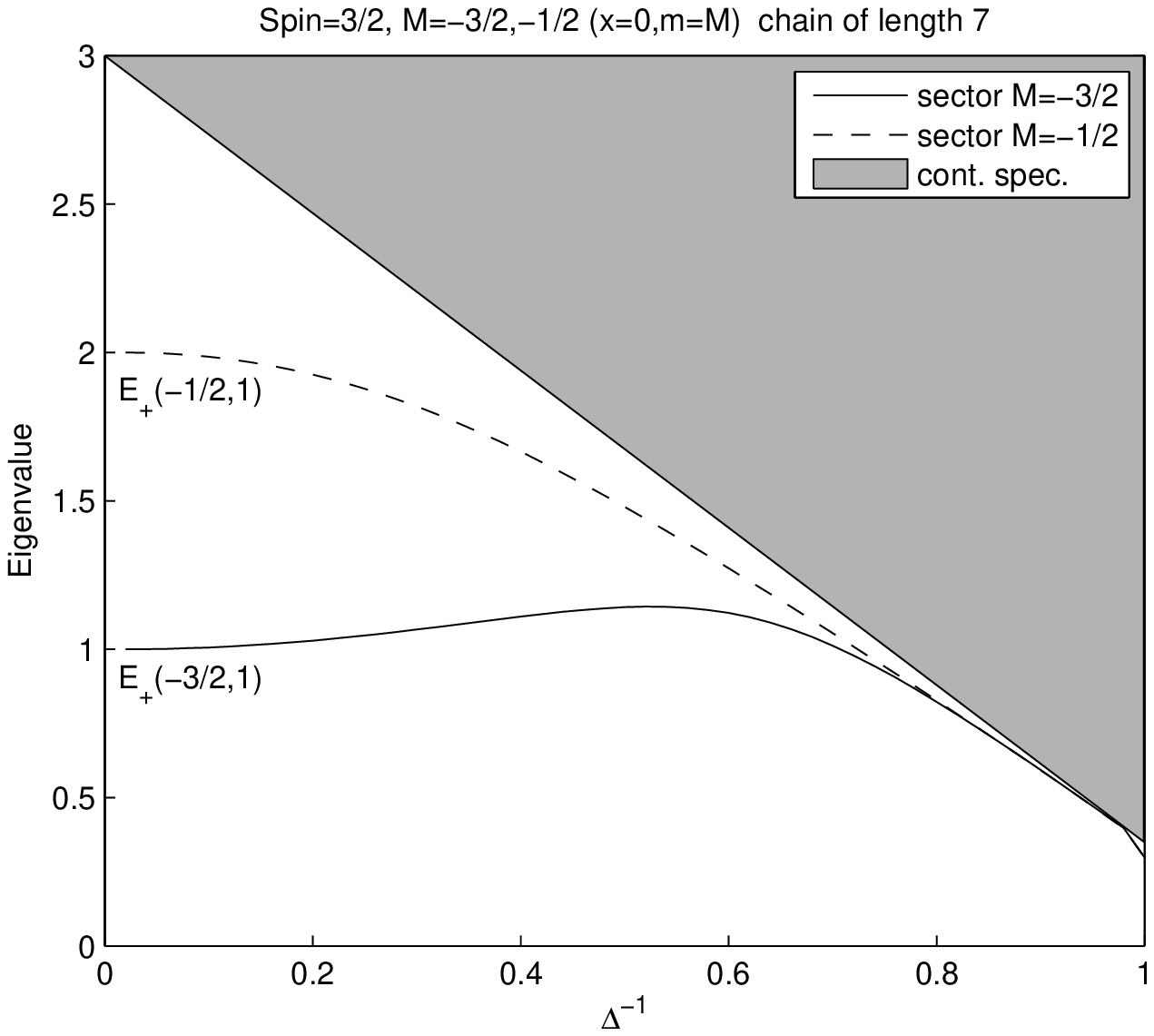,height=6cm,width=6cm}
\epsfig{file=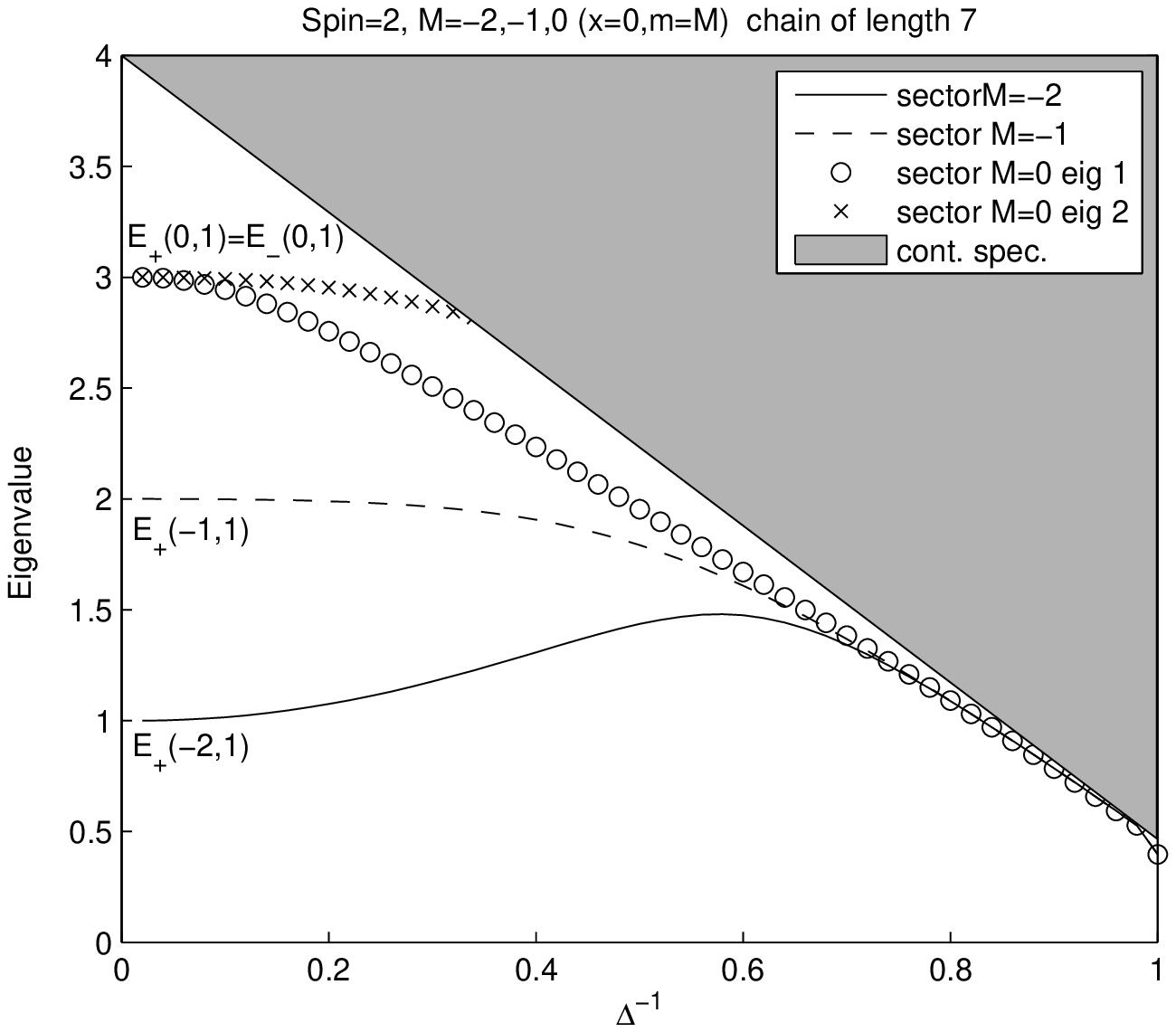,height=6cm,width=6cm}
\caption{The low-lying spectrum of the XXZ model for $J=\frac{3}{2}$ and $J=2$ for various sectors. Note the two-fold degeneracy for spin $2$ in the sector $M=0$ corresponding to $E_+(0,1)=E_-(0,1)$. The shaded region contains those eigenvalues which converge to continuous spectrum in the thermodynamic limit.
We note that, for $\Delta=1$, the gap vanishes in the thermodynamics limit.}
\end{center}
\end{figure}

\Section{Set-up}
We study the spin $J$ ferromagnetic XXZ model on the one-dimensional lattice 
$\mathbb{Z}$. The local Hilbert space for a single site $\alpha$ is $\mathcal{H}_\alpha = \mathbb{C}^{2\J +1}$ with $\Spin \in \frac{1}{2} \mathbb{N} = \{0,\frac{1}{2},1,\frac{3}{2},2,\dots\}$. We consider the Hilbert space for a finite chain on the sites $[-L,L] = \{-L,-L+1,\dots,+L\}$.
This is $\mathcal{H}_{[-L,L]}=\bigotimes_{\alpha=-L}^L\mathcal{H}_\alpha$. 
The Hamiltonian of the spin-$\Spin$ XXZ model is
\begin{equation}
\begin{gathered}
H_L(\Delta^{-1})\, 
=\, \sum_{\alpha=-L}^{L-1} h_{\alpha,\alpha+1}(\Delta^{-1})\, ,\\
h_{\alpha,\alpha+1}(\Delta^{-1})\, 
=\, \Spin^2-S_\alpha^3 S_{\alpha+1}^3 - \Delta^{-1}(S_\alpha^1 S_{\alpha+1}^1 +  S_\alpha^2 S_{\alpha+1}^2)\, ,
\end{gathered}
\end{equation}
where $S_\alpha^1$,$S_\alpha^2$ and $S_\alpha^3$ are the spin-$\Spin$ matrices acting on the site $\alpha$, tensored with the identity operator acting on the other sites.
The main parameter of the model is the anisotropy $\Delta>1$ and we get the Ising limit as $\Delta\rightarrow\infty$. 
It is mathematically more convenient to work with the parameter $\Delta^{-1}$, which we then assume is in the interval $[0,1]$.
As we said, $\Delta^{-1}=0$ is the Ising limit, and $\Delta^{-1}=1$ is the isotropic XXX Heisenberg model.
It was shown \cite{PS1990,ASW1995,GW1995,KN1998}
that additional ground states emerge when we add particular boundary terms.
Examples of this are the kink and antikink Hamiltonians
\begin{align}
H^{\rm k}_L(\Delta^{-1})\ &=\, H_L(\Delta^{-1}) + \Spin \sqrt{1-\Delta^{-2}}\, (S_{-L}^3 - S_L^3)\\
H^{\rm ak}_L(\Delta^{-1})\, &=\, H_L(\Delta^{-1}) - \Spin \sqrt{1-\Delta^{-2}}\, (S_{-L}^3 - S_L^3)
\end{align}
It is easy to see that the \emph{kink} and \emph{anti-kink}
Hamiltonians are unitarily equivalent. 
We will be mainly interested in the \emph{kink} Hamiltonian with $L\ge 2$.
Note that, by a telescoping sum, we can absorb the boundary fields
into the local interactions:
\begin{gather}
\label{XXZ Hamiltonian}
H_L^{\rm k}(\Delta^{-1})\, 
=\, \sum_{\alpha=-L}^{L-1} h^{\rm k}_{\alpha,\alpha+1}(\Delta^{-1})\, ,\\
\notag
h^{\rm k}_{\alpha,\alpha+1}(\Delta^{-1})\, 
=\,  \Spin^2-S_\alpha^3 S_{\alpha+1}^3 - \Delta^{-1}(S_\alpha^1 S_{\alpha+1}^1 +  S_\alpha^2 S_{\alpha+1}^2) +
    \Spin \sqrt{1-\Delta^{-2}}\, (S_{\a}^3 - S_{\a+1}^3)
\end{gather}
The Ising kink Hamiltonian is the result of taking $\Delta\rightarrow\infty$, equivalently setting $\Delta^{-1}=0$,
namely it is
\begin{equation}
\label{Ising Hamiltonian}
\begin{gathered}
H_L^{\rm k}(0)\, 
=\, \sum_{\alpha=-L}^{L-1} h^{\rm k}_{\a,\a+1}(0)\, ,\\
h^{\rm k}_{\a,\a+1}(0)\, 
=\, (\Spin^2-S_\alpha^3 S_{\alpha+1}^3) + \J (S_{\a}^3 - S_{\a+1}^3)\\ 
=\, (\J + S_{\a}^3)(\J-S_{\a+1}^3)\, .
\end{gathered}
\end{equation}

Each of the Hamiltonians introduced above commutes with the total magnetization
$$
S^3_{\rm tot}\, =\, \sum_{\alpha =-L}^L S^3_\alpha\, .
$$
As indicated in the introduction, for each $M \in \{-\Spin (2L+1),\dots,\Spin (2L+1)\}$,
the corresponding sector is defined to be the eigenspace of $S^3_{\rm
  tot}$ with eigenvalue $M$; clearly, these are invariant 
subspaces for all the Hamiltonians introduced above.

The Ising basis is a natural orthonormal basis for $\Hil_{[-L,L]}$.
At each site we have an orthonormal basis of the 
Hilbert space $\mathcal{H}_\alpha$ given by the eigenvectors of
$S^3_\alpha$ and labeled according to their eigenvalues.
We will denote this by $S^3_\alpha \ket{m}_{\a} = m \ket{m}_{\a}$ 
for $m\in [-\Spin,\Spin]$, and $\a \in [-L,L]$. Here, and throughout
the remainder of the paper, we will use the notation 
$[-J,J]$ for the set $\{-J, -J+1, \dots, J\}$ as we have done with
$[-L,L]$.
Finally, there is an orthonormal basis of the entire Hilbert space
consisting of simple tensor product vectors: $\bigotimes_{\a=-L}^{L} \ket{m_{\a}}_{\a}$.

Also recall that the raising and lowering operators are defined such that
\begin{equation*}
S_{\alpha}^+ \ket{m}_{\a}\, =\,
\begin{cases}
\sqrt{\Spin(\Spin+1) - m(m + 1)}\, \ket{m+1}_{\a} & \text{ if $-\Spin \le m \leq \Spin-1$,}\\
0 &\text{if $m = \Spin$;}
\end{cases}
\end{equation*}
\begin{equation*}
S_{\alpha}^-\ket{m}_{\alpha}\, =\,
\begin{cases}
\sqrt{\Spin(\Spin+1) - m(m - 1)}\, \ket{m-1}_{\alpha} & \text{ if $-\Spin+1 \leq m \leq \Spin$,} \\
0 &\text{if $m = -\Spin$.}
\end{cases}
\end{equation*}
A short calculation shows that $S^1_{\a}$ and $S^2_{\a}$ are given by
$S^1_{a} = (S^+_{\a}+S^{-}_{\a})/2$ and $S^2_{a} = (S^+_{\a}-S^{-}_{\a})/2i$, and
$$
S_{\a}^{1} S_{\a+1}^1 + S_{\a}^2 S^2_{\a+1}\,
=\, \frac{1}{2} \left(S_{\a}^+ S_{\a+1}^- + S_{\a}^- S_{\a+1}^+\right)\, .
$$

\Section{Main theorem}
\label{sec:Main Results}
Many of the results in this paper concern the \emph{kink} Hamiltonian given by ~(\ref{XXZ Hamiltonian}). 
We will study it as a perturbation of the Ising Hamiltonian ~(\ref{Ising Hamiltonian})
in the regime $0<\D^{-1}\ll 1$. 
We denote an Ising configuration as $\bm = (m_{\a})_{\a=-L}^{L}$, 
where $m_{\a} \in [-J,J]$ for each $\a$, and the corresponding basis 
vector as
$$
\ket{\bm}\, =\, \bigotimes_{\a=-L}^{L} \ket{m_{\a}}_{\a}\, .
$$
Observe that the Ising \emph{kink} Hamiltonian is diagonal with
respect to this basis, 
\begin{equation}
\label{eq:IsingEnergy}
\begin{gathered}
H^{\rm k}_L(0)\, \ket{\vec{m}}\, =\, E^{\rm k}(\vec{m})\, \ket{\vec{m}}\, ,\quad \text{ where}\\
E^{\rm k}(\vec{m})\, =\, \sum_{\a=-L}^{L-1} e^{\rm k}(m_{\a},m_{\a+1})\quad \text{ and}\\
e^{\rm k}(m_{\a},m_{\a+1})\, =\,  (J + m_{\a})(J-m_{\a+1})\, .
\end{gathered}
\end{equation}
Since each of the $e^{\rm k}(m_{\a}, m_{\a+1})$ are non-negative, it
is easy to see that the ground states of $H^{\rm k}_{L}(0)$ are all of the form
\begin{equation}
\Psi_0(x,m;L)  = \ket{\vec{m}}\, ,\quad \text{ where }\quad
m_{\a}\, =\, \begin{cases} -\J & \text{ for $-L\le \a\le x-1$}\\
m & \text{ for $\a=x$}\\
J & \text{for $x+1\le \a\leq L$,}\end{cases}
\end{equation}
with some $x \in [-L,L]$ and $m \in [-J,J]$. 
Note that the total magnetization corresponding to 
$\Psi_0(x,m;L)$ is $M=-2\Spin x + m$. As we will verify in
Proposition \ref{prop:Ising Groundstates}, it is easy to 
check that these ground states are unique per sector. 
We do point out that there is a slight ambiguity in 
the above labeling scheme, however, 
since $\Psi_0(x-1,-J;L)$ and $\Psi_0(x,J;L)$ coincide. 
Let us consider the following elementary result.

\begin{theorem}
\label{thm:EssentialSpec}
For $J \in \frac{1}{2} \N$, $L\ge 2$ and $M \in \{-J(2L+1),\dots,J(2L+1)\}$, the
eigenspace of $H^{\rm k}_{L}(0)$ corresponding to energy $E=2J$
has dimension at least equal to $2L+1$.
\end{theorem}

Because of this theorem, we see that, in the $L \to \infty$ limit, 
the energy $E=2J$ is essential spectrum; an eigenvalue
with infinite multiplicity. As our interest is in perturbation 
theory, it is natural for us to restrict our attention to energies 
strictly between $0$ and $2J$. We will call the corresponding 
eigenvectors ``low energy excitations''.

We will now describe all the low energy excitations for the Ising kink
Hamiltonian. In order to do so, it is convenient to introduce a 
family of eigenvectors which contains all possible low energy excitations.

Let $x$ be any site away from the boundary, i.e. take $x \in [-L+1, L-1]$, and 
choose $m \in [-J,J]$. These choices specify a sector $M=-2Jx+m$ and
a groundstate $\Psi_0(x,m;L)$. If $m<J$, we define the following
vectors:
For $1 \leq n \leq J-m$ set 
\begin{equation}
\label{eq:psi+}
\Psi_n^{+}(x,m;L)  \, = \, \ket{\vec{m}}\, ,\quad \text{where}\quad
m_{\a}\, =\, \begin{cases} -J & \text{ if $\a\le x-1$,}\\
m+ n & \text{if $\a=x$,}\\
J-n & \text{if $\a = x + 1$,}\\
J & \text{ if $\a \geq x+2$.}
\end{cases}
\end{equation}
If $-J < m$, we define the following
vectors:
For $1 \leq n \leq J+m$ set 
\begin{equation}
\label{eq:psi-}
\Psi_n^{-}(x,m;L)  \, = \, \ket{\vec{m}}\, ,\quad \text{where}\quad
m_{\a}\, =\, \begin{cases} -J & \text{ if $\a\le x-2$,}\\
-J+ n & \text{if $\a=x-1$,}\\
m-n & \text{if $\a = x $,}\\
J & \text{ if $\a \geq x+1$.}
\end{cases}
\end{equation}

\begin{figure}
\begin{center}
\epsfig{file=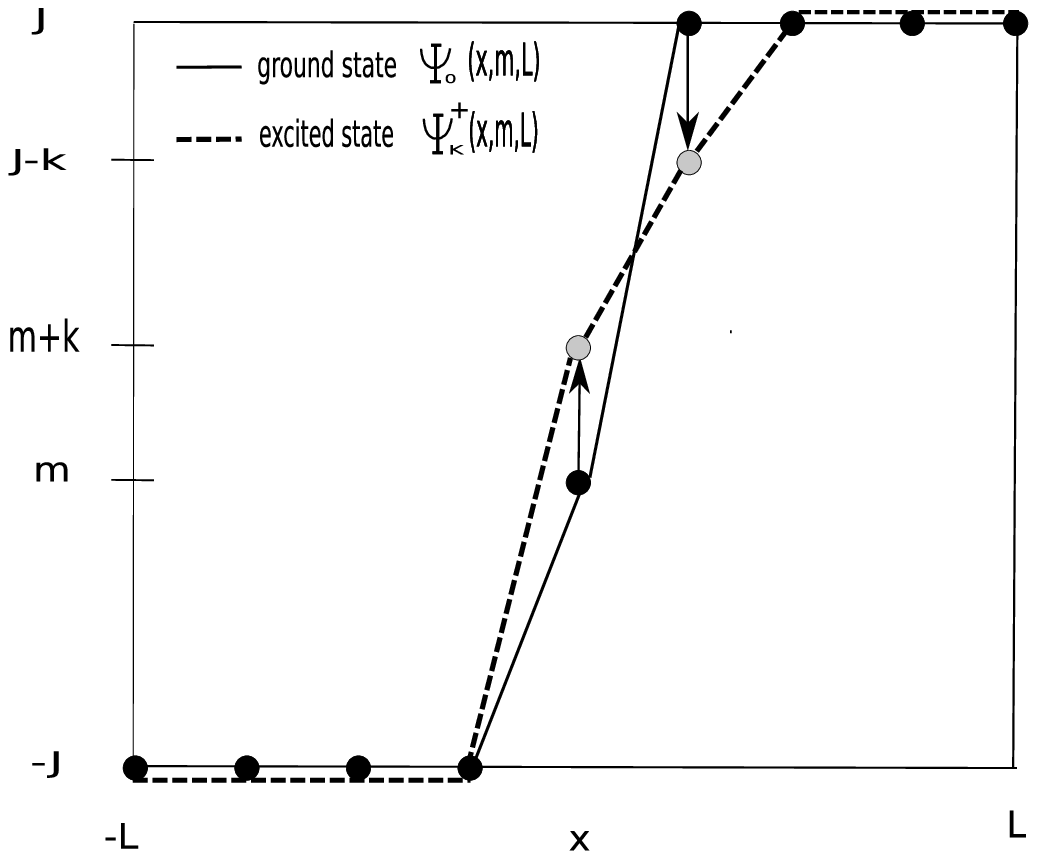,height=5cm,width=7cm}
\epsfig{file=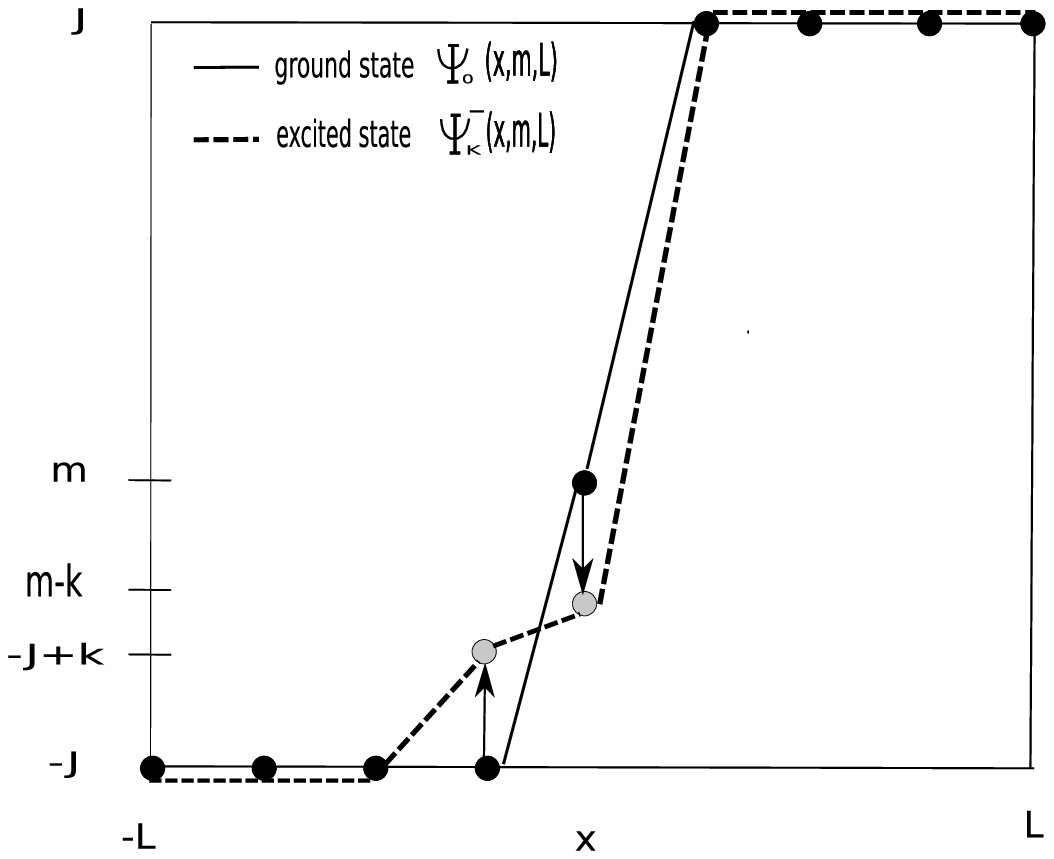,height=5cm,width=7cm}
\caption{\label{fig:Ising_excitations1}
The two kinds of low lying Ising excitations. The left picture 
belongs to the set $K_+(m)$ and the right to the set $K_-(m)$.}
\end{center}
\end{figure}

We call these two sequences of Ising basis vectors the ``localized kink excitations''. Clearly, these vectors $\{ \Psi_n^{\pm}(x,m;L) \}$
have the same total magnetization $M$ and moreover,
\begin{equation}
\label{eq:LocalizedKinkEnergies}
H^{\rm k}_L(0)\, \Psi^{\pm}_n(x,m;L)\, =\, E_{\pm}(m,n)\, \Psi^{\pm}_n(x,m;L),\,
\mbox{with }
E_{\pm}(m,n)\, =\, n^2 + (J\pm m)n\, ,
\end{equation}
where $E_{\pm}(m,n)$ has been calculated using (\ref{eq:IsingEnergy});
note that in each case there is only one non-zero term. In Figure 
\ref{fig:Ising_excitations1} we have shown two ground states, in the left and right graphs, as well as two low excitations. 
These are the schematic diagrams for $\Psi_0(x,m)$, in solid in both pictures, and $\Psi_k^{+}(x,m)$ and $\Psi_k^{-}(x,m)$, in a dashed line, in the left and right, respectively

We now define two sets of labels:
\begin{equation}
\label{eq:Kset}
K_{\pm}(m)\, =\, \{n\, :\, n \in \N\, ,\ 1\leq n\leq J\mp m\, ,\ E_{\pm}(m,n)<2J\}\, .
\end{equation}
Depending on $m$ and $J$, neither, one, or both of these sets may be
empty. We have the following theorem.

\begin{theorem}
\label{thm:Ising Excitations}
(1) The low energy excitations of $H^{\rm k}_{L}(0)$ 
form a subset of the localized kink excitations introduced in
(\ref{eq:psi+}) and (\ref{eq:psi-}) above.\\
(2) For $M$ of the form $M=-2Jx+m$
with some $x \in [-L+1,L-1]$ and $m \in [-J,J]$, the set of low energy excitations equals
$$
\{\Psi^{+}_n(x,m;L)\, :\, n \in K_+(m)\} \cup \{\Psi^{-}_n(x,m;L)\, :\, n \in K_-(m)\}\, .
$$
This is a nonempty set except in the following two cases:
$J=\frac{1}{2}$, or $J =1$, and $m=0$.\\
(3) The low energy excitations of $H_L^k(0)$ are at most 
two-fold degenerate. The first excitation is simple, except for the case that $J$ is an integer $J>1$, and $M = 0$  mod $2J$.
\end{theorem}

\begin{remark}
\label{rem:simpler}
The two-fold degeneracy of the first excited state, guaranteed by part (3) of the theorem, occurs due to the spin flip and reflection symmetry. All other degeneracies with energy $<2J$ occur as follows.
Suppose $2J=ab$ for integers $2\leq a\leq b$.
In this case, let $m=J-2+a-b$.
Then $\Psi_1^+(x,m;L)$ is degenerate with $\Psi_{a-1}^-(x,m;L)$,
both with energy $2J-1+a-b$.
Similarly, $\Psi_1^-(x,-m;L)$ is degenerate with $\Psi_{a-1}^+(x,-m;L)$
for the same energy.
\end{remark}


Next we consider the perturbed Hamiltonian $H^k_L(\Delta^{-1})$. 
As is discussed in Section \ref{sec:perturbation_theory}, each low-lying eigenvalue
$E_\pm(m,n)$, associated to some $n \in K_{\pm}(m)$,
is isolated from the rest of the spectrum by an isolation distance $d_\pm(m,n)
\geq 1$, independent of $L$, which is defined by
$$
d_\pm(m,n)=\inf_{E\in \sigma(H^k_L(0))\setminus \{E_\pm(m,n)\}}
\vert E-E_\pm(m,n)\vert\, .
$$

\begin{theorem}
\label{thm:Main theorem}
Let $L\ge 2$, and fix $J\ge\frac{3}{2}$. For any $n \in K_{\pm}(m)$,
consider the interval
$I=[E_\pm(m,n)-d_\pm(m,n)/2,E_\pm(m,n)+d_\pm(m,n)/2]$ about the low
lying energy $E_{\pm}(m,n)$. The spectral projection of $H_L^{\rm k}(
\Delta^{-1})$ onto $I$ is analytic for large enough values of
$\Delta$. In particular, the dimension of the spectral projection onto
$I$ is constant for this range of $\Delta$. 
\end{theorem}

\begin{remark}
Our estimates yield a lower bound on $\Delta$ as is provided by
(\ref{eq:dbd}). A slightly worse bound demonstrates that taking $\Delta >
18 J^{5/2}$ suffices, but we do not expect either estimate to be sharp.
\end{remark}

The above theorem confirms the structure of the spectrum shown in Figure 
\ref{fig:xxz_spec}.
Moreover, since our numerical calculations indicate that some of the eigenvalues
enter the continous spectrum, we do not expect this type of perturbation 
theory to work for the entire range of $\Delta^{-1}\in [0,1]$. 

\Section{Proof of Theorems ~\ref{thm:EssentialSpec}  and  ~\ref{thm:Ising Excitations} (Excitations of the Ising Model)}
\label{sec:Ising}

In this section, we will focus on the Ising kink Hamiltonian $H_L^{\rm
  k}(0)$, as introduced in (\ref{Ising Hamiltonian}), for a fixed $J \in \frac{1}{2} \N$.
The ground states can be characterized as follows:
\begin{proposition} 
\label{prop:Ising Groundstates}
The ground states of the Ising kink Hamiltonian are all of the form
\begin{equation}
\Psi_0(x,m;L)\, = \ket{\bm}\, ,\quad \text{ where }\quad
m_{\a}\, =\, \begin{cases} -\J & \text{ for $\a=-L,\dots,x-1$,}\\
m & \text{ for $\a=x$,}\\
+\J & \text{ for $\a=x+1,\dots,L$;}\end{cases}
\end{equation}
for $x \in [-L,\dots,L]$ and $m \in [-J,J]$.
Moreover, there is exactly one ground state in each sector;
the total magnetization  eigenvalue for $\Psi_0(x,m;L)$ is $M=2Jx+m$.
\end{proposition}
This proposition was already used, implicitly, in setting up Theorem \ref{thm:Ising Excitations}.
Here we prove it.

\begin{proof}
Given any Ising configuration $\bm = (m_{\a})_{\a=-L}^{L}$, 
equation (\ref{eq:IsingEnergy}) demonstrates that the energy 
associated to $\ket{\bm}$,  $E^{\rm k}(\bm)$, is the sum 
of $2L$ non-negative terms. Therefore, the only way to 
have $E^{\rm k}(\bm)=0$ is if all of the summands are $0$.
This is the case only if either $m_{\a}=-J$ or $m_{\a+1}=+J$ for all $\a$.
Clearly, this is satisfied for the Ising configurations with
$m_{\a}=-J$ for all $\a<x$, $m_{\a}=+J$ for all $\a>x$,
and $m_x$ equal to any number in $[-J,J]$.
It is equally easy to see that these are all of the ground state configurations: if $\bm$ is a groundstate configuration
and for some $x$ we have $m_{x}\neq -J$,
then $m_{x+1}=J$, which, by induction, means that $m_{\a}=+J$ for all $\a>x$.
Similar reasoning yields that $m_{\a}=-J$ for all $\a<x$.

To show that the ground states are unique in each sector,
consider the equation $M-m=-2\Spin x$, subject to the constraints: 
$M \in \{-J(2L+1),\dots,J(2L+1)\}$, $m\in[-J,J]$, and $x\in[-L,L]$.
If $M-J \neq 0$ mod $2J$, then there is a unique pair $(x,m)$ which satisfies this equation.
If $M-J = 0$ mod $2J$, then there are two possible solutions $(x,J)$ and $(x-1,-J)$ for some $x$.
But then, it is trivial to see that $\Psi_0(x,J;L)$ and $\Psi_0(x-1,-J;L)$ coincide.
\end{proof}

Thus in any sector $M$, there is a unique ground state eigenvector 
$\Psi_0(x,m;L)$ for some choice of $(x,m)$ with $M = -2Jx +m$. 
We next observe that there are many eigenvectors with eigenvalue $2J$;
recall this was the statement of Theorem \ref{thm:EssentialSpec}.

\smallskip
\begin{proofof}{Theorem \ref{thm:EssentialSpec}}
First consider the case that $M-J$ is not divisible by $2J$.
Then the unique groundstate eigenvector is $\Psi_0(x,m;L)$ for some $x \in [-L,L]$
and $-J<m<J$. If $x>-L$, then for each $y \in [-L,x-2]$ 
consider the Ising configuration $\vec{m}$ with components
$$
m_{\a}\, 
=\, \begin{cases} -J & \text{ if $\a\le x-1$ and $\a \neq y$;}\\
-J+1 & \text{ if $\a=y$;}\\
m-1 & \text{ if $\a=x$;}\\
J & \text{ if $\a \ge  x+1$.}
\end{cases}
$$
The total magnetization for the vector $\ket{\vec{m}}$ is still $M$.
Using the formula (\ref{eq:IsingEnergy}) again, it is easy to 
check that $e^{\rm k}(m_{y},m_{y+1})=2J$,
$e^{\rm k}(m_{\a},m_{\a+1})=0$ for all $\a \neq y$, and therefore, $E^{\rm k}(\vec{m})=2J$.
This constitutes $x-1+L$ possible values of $y$; producing at least
this many distinct eigenvectors with energy $2J$.
Similarly, if $x<L$ then there are $L-x-1$ Ising configurations of the form $\vec{m}$ with components
$$
m_{\a}\, 
=\, \begin{cases} -J & \text{ if $\a\le x-1$;}\\
m+1 & \text{ if $\a=x$;}\\
J-1 & \text{ if $\a=y$;}\\
J & \text{ if $\a \ge  x+1$ and $\a \neq y$;}
\end{cases}
$$
corresponding to some $y \in [x+2,L]$.
In total, this gives $2L-2$ orthonormal eigenvectors corresponding to eigenvalue $2J$,
yielding a lower bound on  the dimension of the eigenspace.
If $x=-L$ or $L$, then the dimension is increased by at least 1.
In the special case where $M-J$ is divisible by $2J$, the eigenvector
can be written in two ways as $\Psi_0(x,J;L)$ or $\Psi_0(x-1,-J;L)$.
When constructing excitations for $y$ to the left of the kink, use
the first formula above relative to $\Psi_0(x,J;L)$.
When constructing excitations for $y$ to the right of the kink,
use the second formula above relative to $\Psi_0(x-1,-J;L)$.
Once again, this results in $2L-1$ orthonormal eigenvectors
corresponding to eigenvalue $2J$.
\end{proofof}

We now claim that any Ising configuration which is neither a ground
state nor a localized kink excitation corresponds to an energy that is
at least $2J$. This is the content of the following lemma.

\begin{lemma}\label{lemma:Ising Energy}
Consider an Ising configuration $\bm = (m_{\a})_{\a=-L}^{L}$.\\
(1) If there is any $x \in [-L,L-1]$ such that $m_x > m_{x+1}$, then $E^{\rm k}(\bm) \ge 2J$.\\
(2) If there is any $x \in [-L+1,\dots,L-1]$ such that $m_{x-1} > -J$ and $m_{x+1} < J$,
then $E^{\rm k}(\bm) \ge 2J$.
\end{lemma}
\begin{proof}
(1) It is clear from (\ref{eq:IsingEnergy}) that we need only prove 
\begin{equation}
J^2 - m_x m_{x+1} + J(m_x-m_{x+1}) \, = \,  e^{\rm k}(m_x,m_{x+1}) \, \ge \, 2J.
\end{equation}
Since $m_x > m_{x+1}$, we have that $\Spin(m_{x} - m_{x+1}) \geq J$.
We need only verify that $J^2 - m_x m_{x+1} \ge J$
to establish the claim.
The product of two integers $m_x, m_{x+1} \in [-J,J]$ 
is at most equal to $J^2$.
But this is only attained by $m_x=m_{x+1}=\pm J$.
Since $m_x>m_{x+1}$, we can neither have $m_{x+1}=J$ nor $m_{x}=-J$.
The next largest possible value of $m_x m_{x+1}$ is $J(J-1)$, and we
have verified the claim. 

\smallskip
\noindent
(2) Again, because all terms are nonnegative, it suffices to show that 
$$
e^{\rm k}(m_{x-1},m_x) + e^{\rm k}(m_x,m_{x+1})\, \ge\, 2J\, .
$$
Using the formula for $e^{\rm k}(\cdot,\cdot)$ and simplifying gives
$$
e^{\rm k}(m_{x-1},m_x) + e^{\rm k}(m_x,m_{x+1})\, 
=\, 2J^2 + J(m_{x-1}-m_{x+1}) -m_x (m_{x-1} + m_{x+1})\, .
$$
Since $|m_x|\leq J$, we have
\begin{align*}
e^{\rm k}(m_{x-1},m_x) + e^{\rm k}(m_x,m_{x+1})\, 
&\geq\, 2J^2 + J(m_{x-1}-m_{x+1}) - J  |m_{x-1} + m_{x+1}|\\
&=\, J(2J + m_{x-1}-m_{x+1} - |m_{x-1} + m_{x+1}|)\, .
\end{align*}
Therefore,
$$
e^{\rm k}(m_{x-1},m_x) + e^{\rm k}(m_x,m_{x+1})\, 
\geq\, 
\begin{cases}
2J(J-m_{x+1}) & \text{ if $m_{x-1}+m_{x+1}\geq 0$;}\\
2J(J+m_{x-1}) & \text{ if $m_{x-1}+m_{x+1}\leq 0$.}
\end{cases}
$$
Since $J-m_{x+1} \geq 1$
and $J+m_{x-1} \geq 1$, in either case we have proven the claim.
\end{proof}

Now we can finish the proof of the first main result, Theorem \ref{thm:Ising Excitations}.

\begin{proofof}{Theorem \ref{thm:Ising Excitations}}
Part (1) of the theorem is a direct consequence of Lemma \ref{lemma:Ising Energy}
because the only Ising configurations that do not satisfy condition (1) or (2) of that lemma
are the ground state configurations and the localized kink excitations.

To prove part (2), let $M = -2Jx +m$ for some $x \in [-L+1,L-1]$ and
$m \in [-J,J]$. We first consider the case that $J \geq 3/2$. 

If $J$ is sufficiently large and $|m| \leq J-2$, then $n=1 \in
K_{\pm}(m)$ as  
$$
E_{\pm}(m,1)\, =\, 1^2 + (J \pm m) 1\, \leq\, 1 + J + |m|\, <\, 2J\, .
$$
So, in this case, both $K_+(m)$ and $K_{-}(m)$ are nonempty.

Similarly, if $|m| = J-1$, then either
$$
E_+(-J+1,1) \, = \, 1^2 +(J +(-J+1))1 \, = \, 2,
$$
or
$$
E_-(J-1,1) \, = \, 1^2 +(J-(J-1))1 \, = \, 2.
$$
Hence, $1 \in K_+(m) \cup K_-(m)$ if $|m| = J-1$ and $J \geq 3/2$.

Lastly, if $|m| = J$, then either
$$
E_+(-J,1) \, = \, 1^2 +(J +(-J))1 \, = \, 1,
$$
or
$$
E_-(J,1) \, = \, 1^2 +(J-J)1 \, = \, 1.
$$
Hence, $1 \in K_+(m) \cup K_-(m)$ if $|m| = J-1$ and $J \geq 1$.

We have proved (2) in the case that $J \geq 3/2$. Actually, the last
observation above also verifies (2) in the case that $J=1$ and $m =
\pm 1$. 

Finally, if $J=1$ and $m=0$, then $E_{\pm}(0,1) = 1^2 +(1 \pm 0) =2$
and if $J = 1/2$, then $E_{\pm}( \mp 1/2, 1) = 1^2 +(1/2 -1/2) =
1$. In both these cases, the set of kink excitations is empty.
 
We now prove the first part of (3). First, observe that for any $m \in [-J,J]$, $J \pm m
\geq 0$ and therefore with $m$ fixed, both $E_{\pm}(m,n)$ are increasing
functions of $n$ for $n \geq 0$.
Therefore the only degeneracies that can occur for a particular energy $E$
is if $E_+(m,n_1)=E$ and $E_-(m,n_2)=E$ for some integers $n_1$ and $n_2$.
This is obviously at most a two-fold degeneracy.
 

In order to prove the second part of (3), we will simply prove Remark \ref{rem:simpler}.
Without loss of generality, suppose $m\geq 0$.
Then $E_+(m,2)=4+2J+2m>2J$.
So the only possibility for $E_+(m,n)$ to be less than $2J$ is if $n=1$,
which gives the energy $J+m+1$ and an auxiliary condition, namely $J-m\geq 2$.
A degeneracy happens only if there is a $p\geq 1$ such that $E_-(m,p)=E_+(m,1)$.
This means
$$
p^2 + (J-m) p = J+m+1\qquad \Longleftrightarrow\qquad
(p+1)(p+J-m-1)\, =\, 2J\, .
$$
Setting $a=p+1$ and $b=p+J-m-1$ (which satisfies $b\geq a$ because $J-m\geq 2$)
we have exactly the result claimed.
Note that the second part of (3) refers to the first excitation.
For $m\geq 0$, the lowest excitation is $E_-(m,1)$.
This means $p=1$, which implies $m=0$ and therefore $2J$ is even
with $J\geq 2$ (because $J-m\geq 2$).


\end{proofof}

%
%

\Section{Proof of the main theorem}
\label{sec:mainthm}

The goal of this section is to prove Theorem~\ref{thm:Main theorem}. We will do so by
analyzing the kink Hamiltonian $H_L^{\rm k}( \Delta^{-1})$ as a
perturbation of the Ising limit $H_L^{\rm k}(0)$. Within the first
subsection below, specifically in Theorem~\ref{thm:relbd}, we prove that the
operators which arise in our expansion of $H_L^{\rm k}( \Delta^{-1})$
are relatively bounded with respect to $H_L^{\rm k}(0)$. In the
next subsection, we discuss how the explicit bounds on $\Delta$, those 
claimed in Theorem~\ref{thm:Main theorem}, follow from relative boundedness and basic
perturbation theory.

\subsection{Relative Boundedness}
In this subsection, we will analyze the kink Hamiltonian introduced in
(\ref{XXZ Hamiltonian}). Recall that this Hamiltonian is written as
\begin{gather}
H_L^{\rm k}(\Delta^{-1})\, 
=\, \sum_{\alpha=-L}^{L-1} h^{\rm k}_{\alpha,\alpha+1}(\Delta^{-1})\, ,\\
\notag
h^{\rm k}_{\alpha,\alpha+1}(\Delta^{-1})\, 
=\,  \Spin^2-S_\alpha^3 S_{\alpha+1}^3 - \Delta^{-1}(S_\alpha^1 S_{\alpha+1}^1 +  S_\alpha^2 S_{\alpha+1}^2) +
    \Spin \sqrt{1-\Delta^{-2}}\, (S_{\a}^3 - S_{\a+1}^3).
\end{gather}
By adding and subtracting terms of the form $J(S_{\a}^3 -S_{\a +1}^3)$
to the local Hamiltonians, we find that
\begin{equation}
\label{eq:HkinkExpansion}
H_L^{\rm k}( \Delta^{-1}) = H_L^{\rm k}(0) + \Delta^{-1}H_L^{(1)}  +
 \left( 1 - \sqrt{1-\Delta^{-2}} \, \right)H_L^{(2)} 
\end{equation}
where
\begin{equation}
\label{eq:HKinkSummation}
H_L^{(1)}\, =\, - \, \sum_{\a=-L}^{L-1} h_{\a,\a+1}^{(1)} \quad \mbox{
  with } \quad h_{\a,\a+1}^{(1)}\, = \, \frac{1}{2}(S_\alpha^+S_{\alpha+1}^- + S_\alpha^-S_{\alpha+1}^+) 
\end{equation}
and
\begin{equation}
H_L^{(2)} = \Spin(S_{L}^3 - S_{-L}^3).
\end{equation}
In Theorem~\ref{thm:relbd} below, we will show that both $H_L^{(1)}$
and $H_L^{(2)}$ are relatively bounded perturbations of $H_L^{\rm
  k}(0)$. To prove such estimates, we will use the following lemma several
times. 
\begin{lemma}
\label{lem:SA operators}
Let $A$ and $B$ be self-adjoint $n \times n$ matrices.  
If $A \ge 0$ and 
${\rm Ker}(A) \subset {\rm Ker}(B)$, then 
there exists a constant $c>0$ for which, 
\begin{equation}
-c A \, \leq  \, B \, \leq \, c A. 
\end{equation}
One may take $c = \frac{\|B\|}{\lambda_1}$ where $\lambda_1$ denotes the 
smallest positive eigenvalue of $A$.
\end{lemma}

\begin{proof}
Any vector $\psi \in \mathbb{C}^n$ can be written as $\psi = \psi_0 +
\psi_1$ where $\psi_0 \in {\rm Ker}(A)$ and $\psi_1 \in {\rm
  Ker}(A)^{\perp}$. Clearly then, $ \lambda_1 \| \psi_1 \|^2 \leq
\langle \psi, A \psi \rangle$ and therefore
\begin{equation}
\left| \langle \psi, B \psi \rangle \right| \, = \, \left| \langle
  \psi_1 , B \psi_1 \rangle \right| \, \leq \, \| B \| \, \| \psi_1
\|^2 \, \leq \, \frac{\| B \|}{ \lambda_1} \, \langle \psi, A \psi \rangle,
\end{equation}
as claimed.  
\end{proof}

For our proof of Theorem~\ref{thm:relbd}, we find it useful to introduce the 
Ising model without boundary conditions as
an auxillary Hamiltonian, i.e., 
$$
H_L(0)\, =\, \sum_{\alpha=-L}^{L-1} h_{\a,\a+1}(0)\quad \text{where}\quad 
h_{\a,\a+1}(0)\, =\, J^2-S_{\a}^{3} S_{\a+1}^{3}\, .
$$
It is easy to prove the next lemma.
\begin{lemma}
\label{lemma:rb1}
The Ising model without boundary terms is relatively bounded
with respect to the Ising kink Hamiltonian.
In particular, for any vector $\psi$,
$$
\|H_L(0) \psi\|\, \leq\, \|H_L^{\rm k}(0) \psi\| + 2 J^2 \|\psi\|\, .
$$
\end{lemma}
\begin{proof}
Consider the terms of the Ising kink Hamiltonian:
\begin{equation}
h_{\alpha, \alpha +1}^{\rm k}(0) \, = \,
(J+S_{\alpha}^3)(J-S_{\alpha+1}^3) \, = \, h_{\alpha, \alpha+1}(0) +
J(S_{\alpha}^3 - S_{\alpha +1}^3).
\end{equation}
Summing on $\alpha$ then, we find that
\begin{equation}
H_L(0) \, = \, H_L^{\rm k}(0) \, + \, J(S_L^3 - S_{-L}^3), 
\end{equation}
and therefore, the bound
\begin{equation}
\| H_L(0) \psi \| \, \leq \, \|H_L^{\rm k}(0) \psi \| \, + \, 2 J^2 \,
\| \psi \|,
\end{equation}
is clear for any vector $\psi$.
\end{proof}

We now state the relative boundedness result.
\begin{theorem}
\label{thm:relbd}
The linear term in the perturbation expansion of $H_L^{\rm
  k}(\Delta^{-1})$, see (\ref{eq:HkinkExpansion}), satisfies
\begin{equation} \label{eq:relbdp=1}
\left\| H_L^{(1)} \psi \right\| \, \leq \, \sqrt{J^2 \, + \, 2J^3} \,
\|H_L^{\rm k}(0) \psi \| \, + \, 2 J^2 \,  \sqrt{J^2 \, + \, 2J^3}  \,
\| \psi \|,
\end{equation}
for any vector $\psi$. Moreover,  we also have that
\begin{equation} \label{eq:relbdp>1}
\left\| H_L^{(2)} \psi \right\| \, \leq \, 2J^2\left\| \psi \right\|
\end{equation}
\end{theorem}
\begin{proof}
Using Lemma~\ref{lemma:rb1}, it is clear we need only prove that
\begin{equation} \label{lem:hambd}
\left\| H_L^{(1)} \psi \right\| \, \leq \, \sqrt{J^2 \, + \, 2J^3} \left\| H_L(0) \psi \right\|.
\end{equation}
to establish (\ref{eq:relbdp=1}).
To this end, Lemma~\ref{lem:SA operators} provides an immediate bound on the individual
terms of these Hamiltonians. In fact, observe that for any fixed $\alpha$, both $h_{\alpha, \alpha +1}(0)$
and $h_{\alpha, \alpha+1}^{(1)}$ are self-adjoint with $h_{\alpha,
  \alpha+1}(0) \geq 0$ and 
\begin{equation}
{\rm Ker}\left( h_{\alpha, \alpha +1}(0)
\right) \, = \, \left\{ | \vec{m} \rangle \, : \, m_{\alpha} =
  m_{\alpha+1} = \pm J \, \right\} 
\subset {\rm Ker}\left( h_{\alpha, \alpha +1}^{(1)} \right).
\end{equation}
It is also easy to see that, for every $\alpha$, the first positive 
eigenvalue of $h_{\alpha, \alpha   +1}(0)$ is $\lambda_1 = J$, and we have that
\begin{equation}
\| h_{\alpha, \alpha +1}^{(1)} \| \, = \, \frac{1}{2} \left\|
  S_{\alpha}^+ S_{\alpha +1}^- + S_{\alpha}^- S_{\alpha+1}^+ \right\|
\, \leq \, J^2.
\end{equation}
An application of Lemma~\ref{lem:SA operators} yields the operator inequality
\begin{equation} \label{eq:relbd}
-J \, h_{\alpha, \alpha +1}(0) \, \leq \, h_{\alpha, \alpha+1}^{(1)}
\, \leq \, J \, h_{\alpha, \alpha+1}(0),
\end{equation}
valid for any $\alpha$.

The norm bound we seek to prove will follow from considering products
of these local Hamiltonians. For any vector $\psi$, one has that
\begin{equation} \label{eq:normhL1}
\| H_L^{(1)} \psi \|^2 \, = \, \langle \psi, (H_L^{(1)})^2 \psi
\rangle \, = \, \sum_{\alpha, \alpha' = - L}^{L-1} \langle \psi,
h_{\alpha, \alpha+1}^{(1)} \, h_{\alpha', \alpha'+1}^{(1)} \psi \rangle
\end{equation}
and
\begin{equation} \label{eq:normhL0}
\| H_L(0) \psi \|^2 \, = \, \langle \psi, (H_L(0))^2 \psi
\rangle \, = \, \sum_{\alpha, \alpha' = - L}^{L-1} \langle \psi,
h_{\alpha, \alpha+1}(0) \, h_{\alpha', \alpha'+1}(0) \psi \rangle.
\end{equation}

The arguments we provided above apply equally well to the diagonal
terms of (\ref{eq:normhL1}) and (\ref{eq:normhL0}) in the sense that
\begin{equation} \label{eq:relbddia}
-J^2 \, \left(h_{\alpha, \alpha +1}(0) \right)^2 \, \leq \, \left(
  h_{\alpha, \alpha+1}^{(1)} \right)^2
\, \leq \, J^2 \, \left( h_{\alpha, \alpha}(0) \right)^2,
\end{equation}
is also valid for any $\alpha$. We find a similar bound by considering the terms 
on the right hand side of (\ref{eq:normhL1}) and
(\ref{eq:normhL0}) for which $| \alpha - \alpha'| >1$. In this case,
each of the operators $h_{\alpha, \alpha+1}(0)$ and $h_{\alpha,
  \alpha+1}^{(1)}$ commute with both of the operators 
$h_{\alpha', \alpha'+1}(0)$ and $h_{\alpha', \alpha'+1}^{(1)}$.
Moreover, we conclude from (\ref{eq:relbd}) that  the operators 
$g_{\alpha}^{\pm} =  J \, h_{\alpha, \alpha+1}(0)
\pm h_{\alpha, \alpha +1}^{(1)}$ are non-negative for every $\alpha$.
Since all the relevant quantities commute, it is clear that
\begin{equation} \label{eq:gabd}
0 \leq \frac{1}{2} \left( g_{\alpha}^+ g_{\alpha'}^- \, + \,
  g_{\alpha}^- g_{\alpha'}^+ \right) \, = \, J^2 \, h_{\alpha,
  \alpha+1}(0)  h_{\alpha', \alpha'+1}(0) \, - \,  h_{\alpha,
  \alpha+1}^{(1)}  h_{\alpha', \alpha'+1}^{(1)}.
\end{equation}

Our observations above imply the following bound 
\begin{equation} \label{eq:normest}
\left\| H_L^{(1)} \psi \right\|^2 \, - \, J^2 \, \left\| H_L(0) \psi
\right\|^2 \, \leq \, \sum_{\alpha = -L}^{L-2} \left\langle \psi, \left(
 h_{\alpha, \alpha+1}^{(1)} h_{\alpha+1, \alpha+2}^{(1)} \, + \,
 h_{\alpha+1, \alpha+2}^{(1)}  h_{\alpha, \alpha+1}^{(1)}   \right)
\psi \right\rangle.
\end{equation}
In fact, the terms on the right hand side of (\ref{eq:normest}) for
which either $\alpha' = \alpha$ or $| \alpha - \alpha'|>1$ are 
non-positive by (\ref{eq:relbddia}), respectively, (\ref{eq:gabd}).
In the case that $| \alpha - \alpha' | =1$, the operators
$h_{\alpha, \alpha+1}(0) h_{\alpha', \alpha' +1}(0)$ are non-negative
(since they commute) and hence we may drop these terms; those terms that
remain we group as the self-adjoint operators appearing on the right
hand side of (\ref{eq:normest}) above.

Our estimate is completed by applying Lemma~\ref{lem:SA operators} one
more time. Note that for any $\alpha \in \{-L, \cdots, L-2 \}$ the
operator $A_{\alpha} = h_{\alpha, \alpha+1}(0) + h_{\alpha+1, \alpha +2}(0)$ is
self-adjoint, non-negative, and 
\begin{equation}
{\rm Ker}(A_{\alpha}) \, = \, \left\{ | \vec{m} \rangle \, : \, m_{\alpha} \, =
  \,  m_{\alpha+1} \, = \, m_{\alpha+2} \, = \, \pm J \right\} \,
\subset \, {\rm Ker}(B_{\alpha}),
\end{equation}
where the self-adjoint operator $B_{\alpha}$ appearing above is given by
\begin{equation}
B_{\alpha} =  h_{\alpha, \alpha+1}^{(1)} h_{\alpha+1, \alpha+2}^{(1)} \, + \,
 h_{\alpha+1, \alpha+2}^{(1)}  h_{\alpha, \alpha+1}^{(1)}.
\end{equation}
For each $\alpha$, the first positive eigenvalue of $A_{\alpha}$ is
$\lambda_1( \alpha) = J$, and it is also easy to see that $\|
B_{\alpha} \| \leq 2 J^4$. Thus, term by term Lemma~\ref{lem:SA
  operators} implies that 
\begin{equation}
\langle \psi, B_{\alpha} \psi \rangle \, \leq \, 2 J^3 \, \langle
\psi, A_{\alpha} \psi \rangle,
\end{equation}
from which we conclude that 
\begin{equation}
\left\| H_L^{(1)} \psi \right\|^2 \, \leq \, \left( J^2 \, + \, 2 J^3
\right) \, \left\| H_L(0) \psi \right\|^2, 
\end{equation}
as claimed in (\ref{lem:hambd}). We have proved (\ref{eq:relbdp=1}).

Equation  (\ref{eq:relbdp>1}) follows directly from the easy observation that $\left\| H_L^{(2)}\right\|$ is equal to $2J^2$.
\end{proof}

%
%
%

\subsection{Perturbation theory}\label{sec:perturbation_theory}

In Section~\ref{sec:Ising}, we verified that, in any given sector, the spectrum of the
Ising kink Hamiltonian, $H_L^{\rm k}(0)$, when restricted to the
interval $[0, 2J)$ consists of only isolated eigenvalues whose multiplicity
is at most two. In fact, for the sector $M = - 2Jx +m$ these
eigenvalues are determined by 
\begin{equation} \label{eq:evals}
E_{\pm}(m,n) \, = \, n^2 \, + \, (J \, \pm \, m) n
\end{equation}
for those values of $n \in \mathbb{N}$ with $E_{\pm}(m,n) < 2J$. It is
clear from (\ref{eq:evals}) that each of these eigenvalues have an isolation
distance $d_{\pm}(m,n)>0$ from the rest of the spectrum and that this
distance is independent of the length scale $L$. 

For our proof of the relative boundedness result in Theorem~\ref{thm:relbd}, we
expanded the Hamiltonian as
\begin{equation}
H_L^{\rm k}( \Delta^{-1}) \, = \, H_L^{\rm k}(0) \, + \, \Delta^{-1}
H_L^{(1)} \, + \, \left( 1 - \sqrt{1 - \Delta^{-2}} \right) H_L^{(2)}.
\end{equation}
Using the first resolvent formula, it is easy to see that
\begin{equation}
\left( H_L^{\rm k}( \Delta^{-1}) - \xi \right)^{-1} \, = \, R( \xi) \,
\left[ 1 \, + \, \left(  \Delta^{-1}
H_L^{(1)} \, + \, \left( 1 - \sqrt{1 - \Delta^{-2}} \right) H_L^{(2)} \right) \, R( \xi) \, \right]^{-1},
\end{equation}
where we have denoted the resolvent by $R(\xi) = \left( H_L^{\rm k}(0)
  - \xi \right)^{-1}$, and it is assumed that $\Delta^{-1}$ has been
chosen small enough so that
\begin{equation} \label{eq:norm<1}
\left\| \,  \left(  \Delta^{-1}
H_L^{(1)} \, + \, \left( 1 - \sqrt{1 - \Delta^{-2}} \right) H_L^{(2)}
\right) \, R( \xi) \, \right\| \, < \, 1.
\end{equation} 
It is clear from sections II.1.3-4 of \cite{TK1982} and chapter I of \cite{RS1972} that the spectral
projections corresponding to $H_L^{\rm k}( \Delta^{-1})$ can be
written as a power series in $\Delta^{-1}$, the coefficients of which
being integrals of the resolvent over a fixed contour
$\Gamma$. Proving an estimate of the form (\ref{eq:norm<1}) for
$\Delta$ large enough, uniformly with respect to $\xi \in \Gamma$, is 
sufficient to guarantee analyticity of the spectral projections. We
verify such a uniform estimate below. 

Let $E_{\pm}(m,n)$ be an eigenvalue of $H_L^{\rm k}(0)$
with isolation distance $d_{\pm}(m,n)$ as specified above. Denote by
$\Gamma$ the circle in the complex plane centered at $E_{\pm}(m,n)$
with radius $d_{\pm}(m,n)/2$. We claim that if 
\begin{equation}
\Delta \, > \, 18 J^{5/2},
\label{eq:twopiesquared}\end{equation}
then (\ref{eq:norm<1}) is satisfied uniformly for $\xi \in \Gamma$.

We proved in Theorem~\ref{thm:relbd} that for any vector $\psi$,
\begin{equation} \label{eq:relbdlin}
\left\| H_L^{(1)} \psi \right\| \, \leq \, \sqrt{J^2 \, + \, 2J^3} \,
\|H_L^{\rm k}(0) \psi \| \, + \, 2 J^2 \,  \sqrt{J^2 \, + \, 2J^3}  \,
\| \psi \|.
\end{equation}
Applying this bound to vectors $\psi$ of the form $\psi = R( \xi) \phi$
yields a norm estimate on $H_L^{(1)} R(\xi)$, i.e.,
\begin{equation}
\left\| H_L^{(1)} R( \xi) \phi \right\| \, \leq \, \left( \sqrt{J^2 \, + \, 2J^3} \,
\|H_L^{\rm k}(0) R( \xi) \| \, + \, 2 J^2 \,  \sqrt{J^2 \, + \, 2J^3}  \,
\| R( \xi) \| \right) \| \phi \|.
\end{equation}
Moreover, since
\begin{equation}
\|H_L^{\rm k}(0) R( \xi) \| \, \leq \, 1 \, + \, |\xi| \, \| R(\xi) \|,
\end{equation}
we have proved that
\begin{equation}
\left\|  \Delta^{-1} H_L^{(1)} R( \xi) \right\| \, \leq \, \Delta^{-1}
\sqrt{J^2 + 2 J^3} \left[ 1 \, + \, (| \xi| \, + \, 2 J^2 ) \, \| R(\xi)
  \| \, \right].
\end{equation}

Similar arguments, again using Theorem~\ref{thm:relbd}, imply that
\begin{equation} \label{eq:relbdquad}
\left\| \left( 1 - \sqrt{1- \Delta^{-2}} \right) H_L^{(2)} R( \xi)
\right\| \, \leq \, \left( 1 - \sqrt{1 - \Delta^{-2}} \right) \, 2J^2 \, \left\| R(\xi) \right\|.
\end{equation}

For $\xi \in \Gamma$, the circular contour described above, we have that
\begin{equation}
\| R(\xi) \|  \, = \,  \frac{1}{{\rm dist}(\xi,\sigma(H_L^{\rm k}(0)))} \, =
 \, \frac{2}{d_\pm(m,n)},
\end{equation}
and
\begin{equation}
| \xi | \, \leq \, E_{\pm}(m,n) \, + \, d_{\pm}(m,n)/2.
\end{equation}

We derive a bound of the form (\ref{eq:norm<1}), uniform for $\xi \in \Gamma$,
by ensuring $\Delta$ large enough so that
\begin{equation}\label{eq:ineqty}
C_1 \Delta^{-1} \, + \, C_2 (1 - \sqrt{1- \Delta^{-2}}) \, < \, 1, 
\end{equation}
where 
\begin{equation}
C_1 = \sqrt{J^2 + 2 J^3} \left[ 1 \, + \, \left(
  \frac{2E_{\pm}(m,n)}{d_{\pm}(m,n)} \, + \, 1 \, +
  \frac{4J^2}{d_{\pm}(m,n)} \right) \right] \quad \mbox{ and } 
\quad C_2 = \frac{4J^2}{d_{\pm}(m,n)}.
\end{equation}
Explicitly, one finds that the inequality (\ref{eq:ineqty}) is
satisfied for all 
\begin{equation} \label{eq:dbd}
\Delta > \frac{C_1^2+C_2^2}{ C_2\sqrt{C_1^2+2C_2-1}+ C_1 -C_1C_2}
\end{equation}
Equation (\ref{eq:twopiesquared}) is a simple sufficient condition
for $\Delta$ to satisfy this inequality. This is easy to verify if one
first replaces $1-\sqrt{1- \Delta^{-2}}$ by $\Delta^{-2}$ in (\ref{eq:ineqty}).


\Section{Acknowledgements}

This article is based on work supported in part by the U.S. National Science
Foundation under Grant \# DMS-0605342. J.M. wishes to thank NSF Vigre grant \#DMS-0135345. 
The research of S.S.\ was supported in part by a U.S.\ National Science Foundation
grant, \#DMS-0706927.

\end{document}